\documentclass[a4paper]{article}

\usepackage{amsmath,amssymb,amscd,graphicx}
\usepackage{amsthm,color,bm}
\usepackage{array}

\numberwithin{equation}{section}
\newtheorem{theorem}{\bf Theorem}[section]
\newtheorem{lemma}[theorem]{\bf Lemma}

\theoremstyle{remark}

\begin{document}

\title{The Rigid Body Dynamics in an ideal fluid: \\
Clebsch Top and Kummer surfaces}

\author{Jean-Pierre Fran\c{c}oise\thanks{Laboratoire Jacques-Louis Lions, 
Sorbonne-Universit\'{e}, 4 Pl. Jussieu, 75252 
Paris, France.  
E-mail: \texttt{jean-pierre.francoise@upmc.fr}}
\, and\,
Daisuke Tarama\thanks{Daisuke TARAMA, Department of Mathematical Sciences
Ritsumeikan University, 1-1-1 Nojihigashi, Kusatsu, Shiga, 525-8577, Japan
E-mail: \texttt{dtarama@fc.ritsumei.ac.jp}}}
\date{This version: \today \\ \vspace{5mm}
{\it Dedicated to Emma Previato}}

\maketitle
\begin{center}Abstract\end{center}
This is an expository presentation of a completely integrable Hamiltonian system of Clebsch top under a special condition introduced by Weber. 
After a brief account on the geometric setting of the system, the structure of the Poisson commuting first integrals is discussed following the methods by Magri and Skrypnyk. 
Introducing supplementary coordinates, a geometric connection to Kummer surfaces, a typical class of K3 surfaces, is mentioned and also the system is linearized on the Jacobian of a hyperelliptic curve of genus two determined by the system. 
Further some special solutions contained in some vector subspace are discussed. 
Finally, an explicit computation of the action-angle coordinates is introduced. 

\section{Introduction}
This article provides expository explanations on a class of completely integrable Hamiltonian systems describing the rotational motion of a rigid body in an ideal fluid, called Clebsch top. 
Further, some new results are exhibited about the computation of action variables of the system. 

In analytical mechanics, the Hamiltonian systems of rigid bodies are typical fundamental problems, among which one can find completely integrable systems in the sense of Liouville, such as Euler, Lagrange, Kowalevski, Chaplygin, Clebsch, and Steklov tops. 
These completely integrable systems are studied from the view point of geometric mechanics, dynamical systems, (differential) topology, and algebraic geometry. 
A nice survey on the former four systems (Euler, Lagrange, Kowalevski, Chaplygin tops) of completely integrable heavy rigid body is \cite{audin_1996}. 
In particular, the authors have studied some elliptic fibrations arising from free rigid bodies, Euler tops, with connections to Birkhoff normal forms. 
(See \cite{naruki_tarama_2012,tarama_francoise_2014,francoise_tarama_2015}.)

As to Clebsch top, an integrable case of rigid bodies in an ideal fluid, classical studies about the integration of the system were done in \cite{kot}. 
In \cite{kot}, K\"{o}tter has given the integration of the flows on a certain Jacobian variety. 
One of the achievements of \cite{kot} was to identify the intersection of the quadrics with the manifold of lines tangent to two different confocal ellipsoids. 
Modern studies on this topic have been done in \cite{EFe, MaS}. 

Before \cite{kot}, Weber has also worked on the same system under a special condition about the parameters (\cite{web2}), as is also explained briefly in \cite{aom}. 
The same type of integration method was applied for Frahm-Manakov top in \cite{schottky_1891} by Schottky. 
Modern studies on such a type of systems can be found in \cite{admo84,admo87,haine}. 

One of the interesting aspects of Clebsch top under Weber's condition is its relation to Kummer surfaces, a typical class of K3 surfaces. 
In complex algebraic geometry, it is known that Kummer surface is characterized as a quartic surface with the maximal number (i.e. $16$) of double points. 
It is usually defined as the quotient of an Abelian surface by a special involution. 
(See \cite{hudson}.) 
Recall that describing the Abelian surface as the quotient $\mathbb{C}^2/\Lambda$ of $\mathbb{C}^2$ by a lattice $\Lambda$, group-theoretically isomorphic to $\mathbb{Z}^4$, the involution is induced from the one of $\mathbb{C}^2$ defined through $z\mapsto -z$ for all $z\in\mathbb{C}^2$.
(See \cite{barth-hulek-peters-vandeven}.)
As to the real aspect of Kummer surfaces, one can look at the beautiful book \cite{fischer} edited by G. Fischer.  

In the case of Clebsch top, the Abelian surface appears as the Jacobian variety of a hyperelliptic curve of genus two. 
The equations of motion for Clebsch top allow to determine this hyperelliptic curve and the Hamiltonian flow is linearized on the Jacobian variety. 

The rigid body in an ideal fluid, described by Kirchhoff equations, is also studied in \cite{HJL} from the viewpoint of dynamical systems theory. 
In \cite{HJL}, there are given some interesting special solutions which seem similar to the ones appearing in the case of free rigid body dynamics for Euler top. 

An important problem about completely integrable systems in the sense of Liouville is to describe the action-angle coordinates whose existence is guaranteed by Liouville-Arnol'd(-Mineur-Jost) Theorem (cf. \cite{arnold_1989}). 
The problem is also closely related to Birkhoff normal forms, which was studied by the authors in the cases of Euler top (\cite{tarama_francoise_2014,francoise_tarama_2015}). 

\medskip

In the present paper, the geometric setting of Kirchhoff equations is explained and the complete integrability is considered following the methods by \cite{MaS}. 
The connection to Kummer surface is observed directly, but, introducing two supplementary coordinates, the Kummer surface is also connected to the Jacobian variety of a hyperelliptic curve of genus two, with which the integration of the system is carried out. 
Two important aspects of the system from the viewpoint of dynamical systems theory are also discussed: 
some special solutions contained in invariant vector subspaces of dimension three and the computation of action-angle coordinates.

It should be pointed out that there are many studies on the relation to integrable systems and K3 surfaces from general point of view as can be found e.g. in \cite{beauville, markushevich, mukai}. 
However, such connection can be found more concretely in the case of Clebsch top under Weber's condition. 
In this sense, the present paper has the same type of intention as \cite{naruki_tarama_2011}, where the relation between Euler top and Kummer surface is discussed. 

\medskip

The structure of the present paper is as follows: \\
Section 2 gives a brief explanation on Kirchhoff equations and Clebsch condition. 
Kirchhoff equations are Hamilton's equations with respect to Lie-Poisson bracket of $\mathfrak{se}(3)^{\ast}=\left(\mathfrak{so}(3)\ltimes \mathbb{R}^3\right)^{\ast}\equiv \mathbb{R}^3\times \mathbb{R}^3=\mathbb{R}^6$, which has two quadratic Casimir functions $C_1$ and $C_2$. 
In fact, Kirchhoff equations can be restricted to the symplectic leaf, defined as the intersection of level hypersurfaces of $C_1$ and $C_2$. 
In the Clebsch case, there exists an additional constants of motion other than the Hamiltonian and hence the restricted system is completely integrable in the sense of Liouville with two degrees of freedom. 
In the present paper, we mainly focus on the case of Weber where $C_1=0$ and $C_2=1$, with which he was able to integrate Kirchhoff equations by using his methods of theta functions in \cite{web1}. 
Note that the condition $C_1=0$ is a strong simplification, while $C_2=1$ can be assumed without loss of generality because it is only a convenient scaling. 

In Section 3, a change of parameters for Kirchhoff equation in Clebsch case is introduced following the methods of \cite{MaS}, which allows to describe symmetrically the two first integrals other than the Casimir functions $C_1$ and $C_2$, as is observed through the functions $C_3$ and $C_4$. 
The intersection of the four quadrics $C_1=0, C_2=1, C_3=c_3,C_4=c_4$ defines, generically, naturally an eightfold covering of a Kummer's quartic surface. 
Further, one introduces the two supplementary coordinates $(x_1, x_2)$ with which the coordinates of the intersection of the four quadrics $C_1=0, C_2=1, C_3=c_3,C_4=c_4$ are clearly described. 
These coordinates are also useful to linearize the Hamiltonian flows. 

In section 4, a family of invariant subspaces of dimension 3 on which solutions are elliptic curves is discussed (\cite{HJL}). 
Such solutions are contained in the discriminant locus of the family of hyperelliptic curves. 
We miss the special solution with $\bm{p}=0$ where the motion is equivalent to the Euler case of the rigid body, since this is incompatible with Weber's condition $C_2=1$.

Section 5 provides an explicit computation of the actions following the method of \cite{Fran}. 
In Section 6, the conclusion and the future perspectives are mentioned.

\section{Kirchhoff equations and Clebsch condition}
Following \cite{HJL}, we consider Kirchhoff equations 
\begin{equation}\label{KE_K}
\begin{cases}
\displaystyle\frac{\displaystyle\mathsf{d}K_1}{\displaystyle\mathsf{d}t}&=\displaystyle\left(\frac{1}{I_3}-\frac{1}{I_2}\right)K_2K_3+\left(\frac{1}{m_3}-\frac{1}{m_2}\right)p_2p_3, \vspace{1mm}\\
\displaystyle\frac{\displaystyle\mathsf{d}K_2}{\displaystyle\mathsf{d}t}&=\displaystyle\left(\frac{1}{I_1}-\frac{1}{I_3}\right)K_3K_1+\left(\frac{1}{m_1}-\frac{1}{m_3}\right)p_3p_1, \vspace{1mm}\\
\displaystyle\frac{\displaystyle\mathsf{d}K_3}{\displaystyle\mathsf{d}t}&=\displaystyle\left(\frac{1}{I_2}-\frac{1}{I_1}\right)K_1K_2+\left(\frac{1}{m_2}-\frac{1}{m_1}\right)p_1p_2, \vspace{1mm}\\
\end{cases}\\
\end{equation}
\begin{equation}\label{KE_p}
\begin{cases}
\displaystyle\frac{\displaystyle\mathsf{d}p_1}{\displaystyle\mathsf{d}t}&=\displaystyle\frac{1}{I_3}p_2K_3-\frac{1}{I_2}p_3K_2, \vspace{1mm}\\
\displaystyle\frac{\displaystyle\mathsf{d}p_2}{\displaystyle\mathsf{d}t}&=\displaystyle\frac{1}{I_1}p_3K_1-\frac{1}{I_3}p_1K_3, \vspace{1mm}\\
\displaystyle\frac{\displaystyle\mathsf{d}p_3}{\displaystyle\mathsf{d}t}&=\displaystyle\frac{1}{I_2}p_1K_2-\frac{1}{I_1}p_2K_1, \vspace{1mm}\\
\end{cases}
\end{equation}
which describe the rotational motion of a ellipsoidal rigid body in an ideal fluid, where the center of buoyancy and that of gravity for the rigid body are assumed to coincide. 
In \eqref{KE_K}, \eqref{KE_p}, $\left(\bm{K}, \bm{p}\right)=\left(K_1, K_2, K_3, p_1, p_2, p_3\right)\in \mathbb{R}^3\times \mathbb{R}^3\equiv \mathbb{R}^6$ and $I_1, I_2, I_3, m_1, m_2, m_3\in\mathbb{R}$ are parameters of the dynamics, which do not depend on the time $t$. 
Kirchhoff equations \eqref{KE_K}, \eqref{KE_p} are Hamilton's equation for the Hamiltonian $H\left(\bm{K}, \bm{p}\right)=\displaystyle \frac{1}{2}\left(\sum_{\alpha=1}^3\frac{K_{\alpha}^2}{I_{\alpha}}+\sum_{\alpha=1}^3\frac{p_{\alpha}^2}{m_{\alpha}}\right)$ with respect to Lie-Poisson bracket 
\[
\left\{F, G\right\}\left(\bm{K}, \bm{p}\right)
=\left\langle \bm{K}, \nabla_{\bm{K}}F\times \nabla_{\bm{K}}G\right\rangle
+\left\langle \bm{p}, \nabla_{\bm{K}} F\times \nabla_{\bm{p}}G- \nabla_{\bm{K}} G\times \nabla_{\bm{p}}F\right\rangle
\]
on $\mathfrak{se}(3)^{\ast}=\left(\mathfrak{so}(3)\ltimes \mathbb{R}^3\right)^{\ast}\equiv \mathbb{R}^3\times \mathbb{R}^3=\mathbb{R}^6$, where $F, G$ are differentiable functions on $\mathbb{R}^6$, $\langle \cdot, \cdot \rangle$ is the standard Euclidean inner product on $\mathbb{R}^3$, and $\nabla_{\bm{p}} F=\displaystyle \left(\frac{\partial F}{\partial p_1}, \frac{\partial F}{\partial p_2}, \frac{\partial F}{\partial p_3}\right)$, $\nabla_{\bm{K}} F=\displaystyle \left(\frac{\partial F}{\partial K_1}, \frac{\partial F}{\partial K_2}, \frac{\partial F}{\partial K_3}\right)$. 
Note that the Hamiltonian vector field $\Xi_F$ for an arbitrary Hamiltonian $F$ is defined through $\Xi_F\left[G\right]=\left\{F, G\right\}$ for all differentiable functions $G$ on $\mathbb{R}^6$ and can be written as 
\[
\left(\Xi_F\right)_{\left(\bm{K},\bm{p}\right)}=\left(\bm{K}\times \nabla_{\bm{K}}F+\bm{p}\times \nabla_{\bm{p}}F, \bm{p}\times \nabla_{\bm{K}}F\right). 
\]
The two functions $C_1\left(\bm{K}, \bm{p}\right)=\left\langle \bm{K}, \bm{p}\right\rangle=\displaystyle K_1p_1+K_2p_2+K_3p_3$, $C_2\left(\bm{K}, \bm{p}\right)=\left\langle \bm{p}, \bm{p}\right\rangle= p_1^2+p_2^2+p_3^2$ are Casimir functions, i.e. $\left\{C_1, G\right\}\equiv 0$, $\left\{C_2, G\right\}\equiv 0$ for an arbitrary differentiable function $G$ on $\mathbb{R}^6$ and hence they are constant of motion, as well as the Hamiltonian $H$. 

\medskip

We now focus on Clebsch case where the parameters $I_1, I_2, I_3, m_1, m_2, m_3$ satisfy 
\begin{align}\label{clebsch_condition}
&m_1I_1\left(m_2-m_3\right)+m_2I_2\left(m_3-m_1\right)+m_3I_3\left(m_1-m_2\right)= 0, \notag \\
\Longleftrightarrow& \frac{I_2-I_3}{m_1}+\frac{I_3-I_1}{m_2}+\frac{I_1-I_2}{m_3}=0, \\
\Longleftrightarrow& \exists \nu, \nu^{\prime}\in\mathbb{R}\; \text{s.t.}\; \forall \alpha=1,2,3, \; \frac{1}{m_{\alpha}}=\nu+\frac{\nu^{\prime} I_{\alpha}}{I_1I_2I_3}. \notag 
\end{align}
Under this condition, the system \eqref{KE_p}, \eqref{KE_p} admits the fourth constant of motion 
\[
L\left(\bm{K},\bm{p}\right)=-\left(\frac{p_1^2}{I_2I_3}+\frac{p_2^2}{I_3I_1}+\frac{p_3^2}{I_1I_2}\right)+\frac{K_1^2}{m_1I_1}+\frac{K_2^2}{m_2I_2}+\frac{K_3^2}{m_3I_3}.
\]
The four constants of motion $H$, $L$, $C_1$, $C_2$ Poisson commute and hence they define a completely integrable system in the sense of Liouville on each generic coadjoint orbit (symplectic leaf) in $\mathfrak{se}(3)^{\ast}\equiv \mathbb{R}^6$ defined as the intersection of the level hypersurfaces of $C_1$ and $C_2$. 
(See \cite{arnold_1989, marsden-ratiu} for the generalities of completely integrable systems in the sense of Liouville and of Lie-Poisson brackets, although some of their notations do not coincide with ours.)

It is known that the coadjoint orbit $V^4$ defined as the common level set of $C_1$ and $C_2$ is diffeomorphic either to the two-dimensional sphere $S^2$ or to its cotangent bundle $T^{\ast}S^2$. 
In the following, we assume that $C_1=0$ and $C_2=1$, in which case the coadjoint orbit is symplectomorphic to the cotangent bundle $T^{\ast}S^2$ of the two-dimensional unit sphere. 
(See e.g. \cite[\S 7.5]{ratiu-et-al}.) 
On the coadjoint orbit $V^4$ equipped with the orbit symplectic form, we can define the ``momentum mapping'' $J: V^4\ni\left(\bm{K},\bm{p}\right)\mapsto \left(H\left(\bm{K},\bm{p}\right), L\left(\bm{K},\bm{p}\right)\right)\in\mathbb{R}^2$. 
(Later, we use another description of first integrals $C_3$ and $C_4$, but the ``momentum mapping'' $J$ is equivalent to the one $J^{\prime}$ defined as $J^{\prime}\left(\bm{K},\bm{p}\right):=\left(C_3\left(\bm{K},\bm{p}\right), C_4\left(\bm{K},\bm{p}\right)\right)$, since $C_3$ and $C_4$ are linear transformations of $H$ and $L$.) 

Clearly, the level hypersurfaces of the function $H$ are ellipsoids and hence it is compact. 
This means that the fiber $J^{-1}(c)$ of the ``momentum mapping'' $J$ for the value $c\in\mathbb{R}^2$ is compact and hence we are able to use Liouville-Arnol'd(-Mineur-Jost) Theorem \cite{arnold_1989} which implies that regular fibers are (finite disjoint union of) real two-dimensional torus (tori) and around these regular fibers we can take the action-angle coordinates. 

\medskip

In what follows, we mainly work on the complexified system except for Section 5. 
The variables and parameters of Kirchhoff equations are naturally extended to a complex analytic ordinary differential equations on $\mathbb{C}^6$ from the original ones on $\mathbb{R}^6$ and the real quadratic polynomial first integrals $C_1$, $C_2$, $H$, $L$ are also automatically extended to complex quadratic polynomials.  

\section{The algebraic linearization}
In this section, we describe the linearization of the Kirchhoff equations \eqref{KE_p}, \eqref{KE_K} with Clebsch condition \eqref{clebsch_condition}.

Following the lines of the presentation of \cite{MaS}, we consider the intersections of the four quadrics defined as the level hypersurfaces of quadratic functions  
\begin{align}\label{four_first_int}
\begin{cases}
C_1=\displaystyle \sum_{\alpha=1}^3 K_\alpha p_\alpha, \\
C_2=\displaystyle \sum_{\alpha=1}^3 p_{\alpha}^2, \\
C_3=\displaystyle \sum_{\alpha=1}^3 \left\{K_{\alpha}^2+(j_1+j_2+j_3-j_\alpha)p_\alpha^2\right\}, \\
C_4=\displaystyle \sum_{\alpha=1}^3 \left(j_\alpha K_{\alpha}^2+\frac{j_1j_2j_3}{j_\alpha}p_\alpha^2\right)
\end{cases}
\end{align}
in $\left(K_1, K_2, K_3, p_1, p_2, p_3\right)\in \mathbb{C}^6$. 
In what follows,  the parameters $j_1, j_2, j_3\in\mathbb{C}$ are fixed to recover Kirchhoff equations \eqref{KE_K}, \eqref{KE_p}. 
\begin{lemma}
\upshape 
The four functions $C_1, C_2, C_3, C_4$ Poisson commute with respect to Lie-Poisson bracket $\left\{\cdot, \cdot\right\}$. 
\end{lemma}
This is a result of straightforward computations of $\left\{C_3, C_4\right\}=0$. 
Recall that $C_1$ and $C_2$ are Casimir functions for Lie-Poisson bracket $\left\{\cdot, \cdot\right\}$. 

We take the Hamiltonian 
\[
\lambda C_3+\lambda^{\prime} C_4 =\sum_{\alpha=1}^3\left(n_{\alpha} K_{\alpha}^2+n_{\alpha}^{\prime}p_{\alpha}^2\right), 
\]
a member of the linear pencil generated by $C_3$ and $C_4$, as well as its Hamiltonian vector field $\Xi_{\lambda C_3+\lambda^{\prime} C_4}=\lambda \Xi_{C_3}+\lambda^{\prime}\Xi_{C_4}$ with respect to Lie-Poisson bracket $\left\{\cdot, \cdot\right\}$. 
Here, we put 
\[
n_{\alpha}=\lambda +\lambda^{\prime} j_{\alpha},\qquad 
n_{\alpha}^{\prime}=\lambda \left(j_1+j_2+j_3-j_{\alpha}\right)+\lambda^{\prime}\frac{j_1j_2j_3}{j_{\alpha}}.
\]
As is pointed out in \cite{MaS}, we recover the original Kirchhoff equations \eqref{KE_K}, \eqref{KE_p}, when we choose the parameters $\lambda, \lambda^{\prime}$ correctly. 
In this case, we have $n_{\alpha}=\displaystyle \frac{1}{2I_{\alpha}}$ and $n_{\alpha}^{\prime}=\displaystyle \frac{1}{2m_{\alpha}}$. 

The family of commuting Hamiltonians leaves invariant the level sets of the functions $C_1, C_2, C_3, C_4$, which are the symplectic leaves of Poisson space $\left(\mathbb{C}^6, \left\{\cdot, \cdot\right\}\right)$

\subsection{Intersection of the three quadrics $C_2=1, C_3=c_3, C_4=c_4$}
We describe the intersection of the three quadrics $C_2=1, C_3=c_3, C_4=c_4$. 
We first observe that the three equations are linear in $p_1^2,p_2^2,p_3^2$. 
We introduce the matrix 
\begin{equation*}
A=
\begin{pmatrix}
1 & 1 & 1\\
j_2+j_3 & j_1+j_3 & j_1+j_2\\
j_2j_3 & j_3j_1 & j_1j_2\\
\end{pmatrix}
\end{equation*}
and write the conditions of intersection of the three quadrics as:
\begin{equation}\label{linear_equation}
\begin{pmatrix}
1 & 1 & 1\\
j_2+j_3 & j_1+j_3 & j_1+j_2\\
j_2j_3 & j_3j_1 & j_1j_2\\
\end{pmatrix}
\begin{pmatrix}
p_1^2\\
p_2^2\\
p_3^2\\
\end{pmatrix}
=
\begin{pmatrix}
1\\
-(K_1^2+K_2^2+K_3^2)+c_3\\
-(j_1K_1^2+j_2K_2^2+j_3K_3^2)+c_4\\
\end{pmatrix}.
\end{equation}

We further assume the order $j_1<j_2<j_3$ of the parameters and the generic condition
\begin{equation*}
\det A=j_1j_2(j_1-j_2)+j_2j_3(j_2-j_3)+j_1j_3(j_3-j_1)=-(j_1-j_2)(j_2-j_3)(j_3-j_1)\neq 0.
\end{equation*}
Under this condition, we have 
\begin{equation}
A^{-1}=-\frac{1}{\Delta}
\begin{pmatrix}
j_1^2(j_2-j_3) & j_1(j_3-j_2) & j_2-j_3\\
j_2^2(j_3-j_1) & j_2(j_1-j_3) & j_3-j_1\\
j_3^2(j_1-j_2) & j_3(j_2-j_1) & j_1-j_2\\
\end{pmatrix},
\end{equation}
where we set $\Delta:=(j_1-j_2)(j_2-j_3)(j_3-j_1)$, 
and hence the solution to \eqref{linear_equation} is written as 
\begin{align}\label{key}
p_1^2&=l-\frac{1}{\Delta}\left\{(j_1-j_2)(j_2-j_3)K_2^2-(j_2-j_3)(j_3-j_1)K_3^2\right\}, \notag \\
p_2^2&=m-\frac{1}{\Delta}\left\{(j_2-j_3)(j_3-j_1)K_3^2-(j_3-j_1)(j_1-j_2)K_1^2\right\}, \\
p_3^2&=n-\frac{1}{\Delta}\left\{(j_3-j_1)(j_1-j_2)K_1^2-(j_1-j_2)(j_2-j_3)K_2^2\right\}, \notag 
\end{align}
where $(l,m,n)$ are functions of the parameters $(j_1,j_2,j_3$) and of the $(c_3,c_4)$, such that $l+m+n=1$. 

\subsection{The associated Kummer surface}

From the equation and the vanishing condition of the first Casimir, we get 
\begin{align}\label{three_squares}
&\sqrt{l-\frac{1}{\Delta}\left\{(j_1-j_2)(j_2-j_3)K_2^2-(j_2-j_3)(j_3-j_1)K_3^2\right\}}K_1 \notag \\
&\; +\sqrt{m-\frac{1}{\Delta}\left\{(j_2-j_3)(j_3-j_1)K_3^2-(j_3-j_1)(j_1-j_2)K_1^2\right\}}K_2\\
&\;\;+\sqrt{n-\frac{1}{\Delta}\left\{(j_3-j_1)(j_1-j_2)K_1^2-(j_1-j_2)(j_2-j_3)K_2^2\right\}}K_3=0. \notag 
\end{align}
This equation represents an eight-fold covering of a Kummer's quartic surface \cite[\S 54]{hudson} and \cite[p.341]{web2}, given by the correspondence $\left(K_1, K_2, K_3\right)\mapsto \left(K_1^2, K_2^2, K_3^2\right)$. 
We homogenize \eqref{three_squares} by setting $\displaystyle \left(K_1^2=\frac{X_1}{X_4}, K_2^2=\frac{X_2}{X_4}, K_3^2=\frac{X_3}{X_4}\right)$ and by using the parameters 
\begin{align*}
d_1&=-\frac{1}{\Delta}(j_3-j_1)(j_1-j_2)=\frac{1}{j_3-j_2},\\
d_2&=-\frac{1}{\Delta}(j_1-j_2)(j_2-j_3)=\frac{1}{j_1-j_3},\\
d_3&=-\frac{1}{\Delta}(j_2-j_3)(j_3-j_1)=\frac{1}{j_2-j_1}.
\end{align*}
As a result, we have 
\begin{align}\label{define surface}
X_1^2(lX_4+d_2X_2-d_3X_3)^2&+X_2^2(mX_4+d_3X_3-d_1X_1)^2+X_3^2(nX_4+d_1X_1-d_2X_2)^2 \notag \\
&-2X_1X_2(lX_4+d_2X_2-d_3X_3)(mX_4+d_3X_3-d_1X_1) \notag \\
&-2X_1X_3(lX_4+d_2X_2-d_3X_3)(nX_4+d_1X_1-d_2X_2) \notag \\
&-2X_2X_3(mX_4+d_3X_3-d_1X_1)(nX_4+d_1X_1-d_2X_2)=0, 
\end{align}
where $\left(X_1:X_2:X_3:X_4\right)$ are regarded as homogeneous coordinates in $\mathbb{CP}^3$. 
Note that we have the relation $\displaystyle \frac{1}{d_1}+\frac{1}{d_2}+\frac{1}{d_3}=0$. 
As is described in \cite[\S 54]{hudson}, the quartic surface defined through \eqref{three_squares} can be shown to be a Kummer surface. 
Here, we give a brief justification to this. 
\begin{theorem}
The quartic surface defined through \eqref{define surface} in $\mathbb{CP}^3$ has 16 double points and hence it is a Kummer surface. 
\end{theorem}
\begin{proof}
The maximal number of double points in $\mathbb{CP}^3$ is 16 (\cite{kummer,nikulin_1975}). 
Thus, we only need to count the number of double points for the surface defined through \eqref{define surface}. 
Setting $U_1=lX_4+d_2X_2-d_3X_3$, $U_2=mX_4+d_3X_3-d_1X_1$, $U_3=nX_4+d_1X_1-d_2X_2$, we can write down \eqref{define surface} as 
\begin{equation}\label{equation_with_u}
X_1^2U_1^2+X_2^2U_2^2+X_3^2U_3^2-2X_1X_2U_1U_2-2X_2X_3U_2U_3-2X_3X_1U_3U_1=0.
\end{equation}
The singular points of the surface appear at the following 14 points:\\
(1) $\left(X_1:X_2:X_3:X_4\right)=\left(1:0:0:0\right), \left(0:1:0:0\right), \left(0:0:1:0\right), \left(0:0:0:1\right)$; (2) two points defined by $X_i=U_i=0$, $i=1,2,3$; (3) $X_{\alpha}=U_{\beta}=U_{\gamma}=0$, $\alpha=1,2,3$, $\{\alpha, \beta,\gamma\}=\{1,2,3\}$; and (4) $U_1=U_2=U_3=0$. 

We find double points of the surface on each of these lines. 
Recall that a normal form of double points is given as $z_1=z_2=z_3=0$ with the equation $z_1^2+z_2z_3=0$. 
(See \cite{barth-hulek-peters-vandeven}.)
\begin{enumerate}
\item[(1)] Consider the point $\left(X_1: X_2: X_3:X_4\right)=\left(0:0:0:1\right)$, where the linear functions $U_1/X_4$, $U_2/X_4$, $U_3/X_4$ are non-zero. 
These functions are regarded as units in the local ring the polynomials. 
The equation \eqref{equation_with_u} of the surface is written as 
\begin{equation}\label{eq_transform_kummer_1}
\left(\frac{U_1}{X_4}\frac{X_1}{X_4}-\frac{U_2}{X_4}\frac{X_2}{X_4}\right)^2-\frac{U_3}{X_4}\frac{X_3}{X_4}\left(2\frac{U_1}{X_4}\frac{X_1}{X_4}+2\frac{U_2}{X_4}\frac{X_2}{X_4}-\frac{U_3}{X_4}\frac{X_3}{X_4}\right)=0,
\end{equation}
which apparently has a double point at $\displaystyle \frac{X_1}{X_4}=\frac{X_2}{X_4}=\frac{X_3}{X_4}=0$ since three polynomials $\displaystyle \frac{U_1}{X_4}\frac{X_1}{X_4}-\frac{U_2}{X_4}\frac{X_2}{X_4}$, $\displaystyle \frac{U_3}{X_4}\frac{X_3}{X_4}$, $\displaystyle 2\frac{U_1}{X_4}\frac{X_1}{X_4}+2\frac{U_2}{X_4}\frac{X_2}{X_4}-\frac{U_3}{X_4}\frac{X_3}{X_4}$ are of weight one. 

Concerning the point $\left(X_1: X_2: X_3:X_4\right)=\left(0:0:1:0\right)$, we see that $U_1\neq 0$, $U_2\neq 0$, $U_3=0$ and hence the polynomials $\displaystyle \frac{U_1}{X_3}\frac{X_1}{X_3}-\frac{U_2}{X_3}\frac{X_2}{X_3}$, $U_3/X_3$, $\displaystyle 2\frac{U_1}{X_3}\frac{X_1}{X_3}+2\frac{U_2}{X_3}\frac{X_2}{X_3}-\frac{U_3}{X_3}$ are of weight one. 
Thus, the transformation 
\begin{equation}\label{eq_transform_kummer_2}
\left(\frac{U_1}{X_3}\frac{X_1}{X_3}-\frac{U_2}{X_3}\frac{X_2}{X_3}\right)^2-\frac{U_3}{X_3}\left(2\frac{U_1}{X_3}\frac{X_1}{X_3}+2\frac{U_2}{X_3}\frac{X_2}{X_3}-\frac{U_3}{X_3}\right)=0
\end{equation}
of \eqref{equation_with_u} indicates that $\left(X_1: X_2: X_3:X_4\right)=\left(0:0:1:0\right)$ is a double point of the surface. 
The argument can also be applied to the other two points $\left(X_1: X_2: X_3:X_4\right)=\left(0:1:0:0\right)$ and $\left(X_1: X_2: X_3:X_4\right)=\left(1:0:0:0\right)$. 
We have in total four double points of this type. 

\item[(2)] 
Consider the points of the surface on the line $X_3=U_3=0$. 
Here, we can assume $X_4\neq 0$, since if $X_4=0$, we have $X_1X_2=0$ by \eqref{equation_with_u} and hence it reduces to the case (1). 
Under this assumption, we have $\displaystyle U_1X_1-U_2X_2=0$ by \eqref{equation_with_u}, which consists of two distinct points, as the condition is reduced to the quadratic equation in $(X_1: X_2)$: $ld_1X_1^2-\left\{(m+n)d_1+(n+l)d_2\right\}X_1X_2+md_2X_2^2=0$. 
Around each of these two points, the polynomials $\displaystyle \frac{U_3}{X_4}$, $\displaystyle \frac{X_3}{X_4}$, $\displaystyle \frac{U_1}{X_4}\frac{X_1}{X_4}-\frac{U_2}{X_4}\frac{X_2}{X_4}$ are of weight one, while $\displaystyle 2\frac{U_1}{X_4}\frac{X_1}{X_4}+2\frac{U_2}{X_4}\frac{X_2}{X_4}-\frac{U_3}{X_4}\frac{X_3}{X_4}$ is a unit. 
Thus, the transformation \eqref{eq_transform_kummer_1} of \eqref{equation_with_u} indicates that each of the two points is double point. 
The same method can be applied to the points where $X_1=U_1=0$ and $X_2=U_2=0$. 
The total number of these double points is six. 

\item[(3)] 
Consider the points with the condition $U_1=U_2=X_3=0$, with which we have $U_3\neq0$, $X_1\neq0$, $X_2\neq 0$, $X_4\neq0$. 
Around this points, the polynomials $\displaystyle \frac{U_1}{X_4}\frac{X_1}{X_4}-\frac{U_2}{X_4}\frac{X_2}{X_4}$, $\displaystyle 2\frac{U_1}{X_4}\frac{X_1}{X_4}+2\frac{U_2}{X_4}\frac{X_2}{X_4}-\frac{U_3}{X_4}\frac{X_3}{X_4}$, $\displaystyle \frac{X_3}{X_4}$ are of weight one, while $\displaystyle \frac{U_3}{X_4}$ is a unit. 
Thus, from the transformation \eqref{eq_transform_kummer_1} of \eqref{equation_with_u}, we see that the point is a double point.
The method is applied also to the points $U_2=U_3=X_1=0$, $U_3=U_1=X_2=0$. 
In total, we have three double points. 

\item[(4)]
At the point $U_1=U_2=U_3=0$, we have $0=U_1+U_2+U_3=\left(l+m+n\right)X_4=X_4$ and hence we obtain $\displaystyle \left(X_1: X_2: X_3:X_4\right)=\left(\frac{1}{d_1}:\frac{1}{d_2}:\frac{1}{d_3}:0\right)$. 
In particular, $X_1$, $X_2$, $X_3$ are non-zero at this point.
The functions $\displaystyle \frac{U_1}{X_3}\frac{X_1}{X_3}-\frac{U_2}{X_3}\frac{X_2}{X_3}$, $\displaystyle \frac{U_3}{X_3}$, $\displaystyle 2\frac{U_1}{X_3}\frac{X_1}{X_3}+2\frac{U_2}{X_3}\frac{X_2}{X_3}-\frac{U_3}{X_3}$ are polynomials of weight one. 
So, again, we have a double point because of the transformation \eqref{eq_transform_kummer_2} of \eqref{equation_with_u}. 
\end{enumerate}
The other two points can be detected as in \cite[\S 55]{hudson}. 
\end{proof}
Note that we obtain Kummer surface depending both on the parameters $(j_1,j_2,j_3)$ and on the values of the two first integrals $(C_3,C_4)$. 

To summarize, there are 

\begin{enumerate}

\item[(1)] The origin in the affine space $[0:0:0:1]$.

\item[(2)] Four points at infinity ($X_4=0$) which are: $[1:0:0:0],[0:1:0:0],[0:0:1:0]$ real independent of any constant and 
$\displaystyle\left[\frac{1}{d_1}:\frac{1}{d_2}:\frac{1}{d_3}:0\right]$ real depending only of the parameters $(j_1,j_2,j_3)$. 
There are two other singular points at infinity ($X_4=0$) obtained through the method in \cite[\S 55]{hudson}. 

\item[(3)] Three points $\displaystyle\left[\frac{m}{d_1}:-\frac{l}{d_2}:0:1\right], \left[0:\frac{n}{d_2}:-\frac{m}{d_3}:1\right], 
\left[-\frac{n}{d_1}:0:\frac{l}{d_3}:1\right]$ in the  real affine space.

\item[(4)] Six points of type $X_3=U_3=0$ (and cyclic perturbations) which are affine either real or complex conjugated depending on the sign of $[(m+n)d_1+(n+l)d_2]^2-4lmd_1d_2$ (and cyclic perturbations).
\end{enumerate}

\subsection{Supplementary coordinates $(x_1,x_2)$}
Inspired by \cite{web2}, we introduce two supplementary coordinates $(x_1,x_2)$ as the solutions of the equation
\begin{equation}
\frac{p_1^2}{x-j_1}+\frac{p_2^2}{x-j_2}+\frac{p_3^2}{x-j_3}=0,
\end{equation}
which can also be written as
\begin{equation}
x^2-Ex+F=0,
\end{equation}
where $\displaystyle E=\sum_{\alpha=1}^3 (j_1+j_2+j_3-j_\alpha)p_\alpha^2$, $\displaystyle F=\sum_{\alpha=1}^3 \frac{j_1j_2j_3}{j_\alpha}p_\alpha^2$.
Note that $(x_1,x_2)$ are functions only of the $p_\alpha, \alpha=1,2,3$, hence $\{x_1,x_2\}=0$ and that the expression of the squared coordinates $p_\alpha^2$, $\alpha=1,2,3$, in terms of $(x_1,x_2)$: 
\begin{equation}\label{alpha}
\begin{cases}
p_1^2=\displaystyle \frac{(j_1-x_1)(j_1-x_2)}{(j_1-j_2)(j_1-j_3)},\\
p_2^2=\displaystyle \frac{(j_2-x_1)(j_2-x_2)}{(j_2-j_3)(j_2-j_1)},\\
p_3^2=\displaystyle \frac{(j_3-x_1)(j_3-x_2)}{(j_3-j_1)(j_3-j_2)},
\end{cases}
\end{equation}
where we used $C_2=p_1^2+p_2^2+p_3^2=1$. 

We also introduce the two solutions $j_4, j_5$ of the equation 
\begin{align*}
\frac{l}{x-j_1}+\frac{m}{x-j_2}+\frac{n}{x-j_3}&=0 \\
\Longleftrightarrow 
x^2-\left\{(j_2+j_3)l+(j_1+j_3)m+(j_2+j_1)n\right\}x&+\left(j_2j_3l+j_1j_3m+j_1j_2n\right)=0.
\end{align*}
By $l+m+n=1$, we have the expression of the parameters $l, m, n$ in terms of $j_1,j_2,j_3$ as 
\begin{equation}\label{beta}
\begin{cases}
l&=\displaystyle \frac{(j_1-j_4)(j_1-j_5)}{(j_1-j_2)(j_1-j_3)}, \\
m&=\displaystyle \frac{(j_2-j_4)(j_2-j_5)}{(j_2-j_3)(j_2-j_1)}, \\
n&=\displaystyle \frac{(j_3-j_4)(j_3-j_5)}{(j_3-j_1)(j_3-j_2)}.
\end{cases}
\end{equation}

Putting $K_1=K_2=K_3=0$ in \eqref{linear_equation}, we have 
\begin{equation}
\begin{pmatrix}
l\\
m\\
n\\
\end{pmatrix}
=A^{-1}
\begin{pmatrix}
1\\
C_3\\
C_4\\
\end{pmatrix},
\end{equation}
which yields 
\begin{equation*}
\begin{cases}
l&=\displaystyle -\frac{1}{\Delta}\left\{j_1^2(j_2-j_3)+C_3j_1(j_3-j_2)+C_4(j_2-j_3)\right\}, \vspace{1mm} \\
m&=\displaystyle -\frac{1}{\Delta}\left\{j_2^2(j_3-j_1)+C_3j_2(j_1-j_3)+C_4(j_3-j_1)\right\}, \vspace{1mm} \\
n&=\displaystyle -\frac{1}{\Delta}\left\{j_3^2(j_1-j_2)+C_3j_3(j_2-j_1)+C_4(j_1-j_2)\right\}.
\end{cases}
\end{equation*}
Therefore, we have 
\begin{align*}
(j_2+j_3)l+(j_1+j_3)m+(j_2+j_1)n
&=-\frac{C_3}{\Delta}\left\{j_1(j_3^2-j_2^2)+j_2(j_1^2-j_3^2)+j_3(j_2^2-j_1^2)\right\}
=C_3, \\
j_2j_3l+j_1j_3m+j_1j_2n
&=-\frac{C_4}{\Delta}[j_2j_3(j_2-j_3)+j_1j_3(j_3-j_1)+j_1j_2(j_1-j_2)]
=C_4.
\end{align*}
and hence $j_4,j_5$ are the two roots of the equation 
\begin{equation}\label{equation_c_3_c_4}
x^2-C_3x+C_4=0. 
\end{equation}
Thus, they are independent of the parameters $j_1,j_2,j_3$ and only depends on $C_3, C_4$. 

\vskip 1pt
We observe that the two vectors 
$$e_1=\begin{pmatrix}\frac{p_1}{x_1-j_1} \\ \frac{p_2}{x_1-j_2} \\ \frac{p_3}{x_1-j_3}\end{pmatrix}, \qquad e_2=\begin{pmatrix}\frac{p_1}{x_2-j_1} \\ \frac{p_2}{x_2-j_2} \\ \frac{p_3}{x_2-j_3}\end{pmatrix}$$ are orthogonal to the vector $(p_1,p_2,p_3)^{\mathrm{T}}$ with respect to the standard Euclidean metric. 
In general, these two vectors are independent. 
The vector $(K_1,K_2,K_3)^{\mathrm{T}}$, which is also orthogonal to $(p_1,p_2,p_3)^{\mathrm{T}}$, is thus a linear combination of these two vectors $(K_1,K_2,K_3)^{\mathrm{T}}=Ae_1+Be_2$, $A, B\in\mathbb{R}$. 
This yields
\begin{equation}\label{defAB}
\begin{cases}
K_1&=\displaystyle \frac{1}{\sqrt{(j_1-j_2)(j_1-j_3)}}\left(A\sqrt{\frac{x_1-j_1}{x_2-j_1}}+B\sqrt{\frac{x_2-j_1}{x_1-j_1}}\right), \\
K_2&=\displaystyle \frac{1}{\sqrt{(j_2-j_3)(j_2-j_1)}}\left(A\sqrt{\frac{x_1-j_2}{x_2-j_2}}+B\sqrt{\frac{x_2-j_2}{x_1-j_2}}\right), \\
K_3&=\displaystyle \frac{1}{\sqrt{(j_3-j_1)(j_3-j_2)}}\left(A\sqrt{\frac{x_1-j_3}{x_2-j_3}}+B\sqrt{\frac{x_2-j_3}{x_1-j_3}}\right).
\end{cases}
\end{equation}

Next,  we insert these formulae inside (for instance) the first equations of $\eqref{key}$ and get 
\begin{align}\label{explicite p}
&\frac{(j_1-x_1)(j_1-x_2)}{(j_1-j_2)(j_1-j_3)}-\frac{(j_1-j_4)(j_1-j_5)}{(j_1-j_2)(j_1-j_3)} \notag \\
& =-\frac{1}{\Delta}\left[-[\left\{A^2(\frac{x_1-j_2}{x_2-j_2})+B^2(\frac{x_2-j_2}{x_1-j_2})-2AB\right\}+
\left\{A^2(\frac{x_1-j_3}{x_2-j_3})+B^2(\frac{x_2-j_3}{x_1-j_3})-2AB\right\}\right]  \notag \\
&=-\frac{1}{\Delta}[-A^2(\frac{x_1-j_2}{x_2-j_2}-\frac{x_1-j_3}{x_2-j_3})-B^2(\frac{x_2-j_2}{x_1-j_2}-\frac{x_2-j_3}{x_1-j_3})] \notag \\
&=-\frac{1}{\Delta}[-A^2(\frac{(x_2-x_1)(j_3-j_2)}{(x_2-j_2)(x_2-j_3)})-B^2(\frac{(x_1-x_2)(j_3-j_2)}{(x_1-j_2)(x_1-j_3)})], 
\end{align}
which yields 
\begin{equation}\label{bis_1}
(j_1-x_1)(j_1-x_2)-(j_1-j_4)(j_1-j_5)=-(x_1-x_2)\left\{\frac{A^2}{(x_2-j_2)(x_2-j_3)}-\frac{B^2}{(x_1-j_2)(x_1-j_3)}\right\}.
\end{equation}
Next, from the formulae \eqref{defAB} and the second equation of \eqref{key}, we obtain 
\begin{equation}\label{bis_2}
(j_2-x_1)(j_2-x_2)-(j_2-j_4)(j_2-j_5)=-(x_1-x_2)\left\{\frac{A^2}{(x_2-j_3)(x_2-j_1)}-\frac{B^2}{(x_1-j_3)(x_1-j_1)}\right\}. 
\end{equation}

We further assume the genericity assumption 
\begin{align*}
\det &\begin{pmatrix}
\frac{1}{(x_2-j_2)(x_2-j_3)} & -\frac{1}{(x_1-j_2)(x_1-j_3)} \\
\frac{1}{(x_2-j_2)(x_2-j_3)} & -\frac{1}{(x_1-j_2)(x_1-j_3)}
\end{pmatrix} \\
&=-\frac{(j_1-j_2)(x_1-x_2)}{(x_1-j_1)(x_1-j_2)(x_1-j_3)(x_2-j_1)(x_2-j_2)(x_2-j_3)}\neq 0
\end{align*}
on $(x_1,x_2, j_1, j_2, j_3)$ with which the two last equations \eqref{bis_1} and \eqref{bis_2} determine uniquely $A^2$ and $B^2$ in terms of $(x_1,x_2, j_1, j_2, j_3)$. 
Then, we have 
\begin{align}\label{a_b_squared}
A^2&=-\frac{\Phi(x_2)\Psi(x_1)}{(x_2-x_1)^2}, \qquad 
B^2=-\frac{\Phi(x_1)\Psi(x_2)}{(x_2-x_1)^2}, \notag \\
\Phi(x)&=(j_1-x)(j_2-x)(j_3-x), \qquad 
\Psi(x)=(x-j_4)(x-j_5). 
\end{align}
This also follows from the following identity satisfied by an arbitrary $\lambda$
\begin{equation}
(\lambda-x_1)(\lambda-x_2)-(\lambda-j_4)(\lambda-j_5)=\frac{1}{(x_2-x_1)}\left\{\Psi(x_1)(x_2-\lambda)-\Psi(x_2)(x_1-\lambda)\right\}.
\end{equation}
Substituting $\lambda$ by $j_1$ and $j_2$, we have \eqref{a_b_squared}. 

Next, from \eqref{defAB} and \eqref{a_b_squared}, we obtain the formula $(7)$ of \cite{web2}:
\begin{align}\label{explicit}
K_1&=\sqrt{\frac{-1}{(j_1-j_2)(j_1-j_3)}}\frac{1}{x_2-x_1}
\left\{\sqrt{\frac{\Phi(x_2)\Psi(x_1)(x_1-j_1)}{x_2-j_1}}-\sqrt{\frac{\Phi(x_1)\Psi(x_2)(x_2-j_1)}{x_1-j_1}}\right\}, \notag \\
K_2&= \sqrt{\frac{-1}{(j_2-j_3)(j_2-j_1)}}\frac{1}{x_2-x_1}\left\{\sqrt{\frac{\Phi(x_2)\Psi(x_1)(x_1-j_2)}{x_2-j_2}}-\sqrt{\frac{\Phi(x_1)\Psi(x_2)(x_2-j_2)}{x_1-j_2}}\right\}, \\
K_3&= \sqrt{\frac{-1}{(j_3-j_1)(j_3-j_2)}}\frac{1}{x_2-x_1}\left\{\sqrt{\frac{\Phi(x_2)\Psi(x_1)(x_1-j_3)}{x_2-j_3}}-\sqrt{\frac{\Phi(x_1)\Psi(x_2)(x_2-j_3)}{x_1-j_3}}\right\}. \notag 
\end{align}

At this point, we can observe that these equations induce an algebraic parametrization of the Kummer surface with the coordinates $(x_1,x_2)$. See also \cite[\S 99]{hudson}.

\subsection{Linearization of the Hamiltonian flows}
\begin{theorem}
The Hamiltonian flows associated with the Hamiltonians $\lambda C_3+\lambda' C_4$ induce linear flows on the Jacobian of the hyperelliptic curves $\mathcal{C}_c$ defined through the equation $y^2=(j_1-x)(j_2-x)(j_3-x)(j_4-x)(j_5-x)$.
\end{theorem}
\begin{proof}

The equations \eqref{KE_p} can be written as:
\begin{align}\label{flowbis}
\begin{cases}
\dot{p}_1&=(\lambda+\lambda^{\prime}j_2)K_2 p_3-(\lambda+\lambda^{\prime}j_3)K_3 p_2, \\
\dot{p}_2&=(\lambda+\lambda^{\prime}j_3)K_3 p_1-(\lambda+\lambda^{\prime}j_1)K_1 p_3, \\
\dot{p}_3&=(\lambda+\lambda^{\prime}j_1)K_1 p_2-(\lambda+\lambda^{\prime}j_2)K_2 p_1. 
\end{cases}
\end{align}
Recall that $n_{\alpha}=\displaystyle \frac{1}{2I_{\alpha}}$ and $n_{\alpha}^{\prime}=\displaystyle \frac{1}{2m_{\alpha}}$.

From \eqref{alpha}, we obtain 
\begin{equation*}
\frac{\mathsf{d}p_1}{\mathsf{d}t}=-\frac{1}{2\sqrt{(j_1-j_2)(j_1-j_3)}}\left(\sqrt{\frac{j_1-x_2}{j_1-x_1}}\frac{\mathsf{d}x_1}{\mathsf{d}t}+\sqrt{\frac{j_1-x_1}{j_1-x_2}}\frac{\mathsf{d}x_2}{\mathsf{d}t}\right).
\end{equation*}

On the other hand, from \eqref{flowbis}, \eqref{explicit}, \eqref{alpha}, we have 
\begin{align*}
\dot{p}_1&= (\lambda+\lambda^{\prime}j_2)\sqrt{\frac{-1}{(j_2-j_3)(j_2-j_1)}}\frac{1}{x_2-x_1} \\
&\quad \times \left\{\sqrt{\Phi(x_2)\Psi(x_1)}\sqrt{\frac{x_1-j_2}{x_2-j_2}}-\sqrt{\Phi(x_1)\Psi(x_2)}\sqrt{\frac{x_2-j_2}{x_1-j_2}}\right\}\sqrt{\frac{(j_3-x_1)(j_3-x_2)}{(j_3-j_1)(j_3-j_2)}}\\
&-(\lambda+\lambda^{\prime}j_3)\sqrt{\frac{-1}{(j_3-j_1)(j_3-j_2)}}\frac{1}{x_2-x_1}\\
&\quad \times \left\{\sqrt{\Phi(x_2)\Psi(x_1)}\sqrt{\frac{x_1-j_3}{x_2-j_3}}-\sqrt{\Phi(x_1)\Psi(x_2)}\sqrt{\frac{x_2-j_3}{x_1-j_3}}\right\}\sqrt{\frac{(j_2-x_1)(j_2-x_2)}{(j_2-j_3)(j_2-j_1)}} 
\end{align*}
\begin{align*}
&= (\lambda+\lambda^{\prime}j_2)\sqrt{\frac{-1}{(j_2-j_3)(j_2-j_1)}}\frac{1}{x_2-x_1}\frac{1}{\sqrt{(j_3-j_1)(j_3-j_2)}} \\
&\quad \times \left\{R(x_1)(x_2-j_3)\sqrt{\frac{x_2-j_1}{x_1-j_1}}-R(x_2)(x_1-j_3)\sqrt{\frac{x_1-j_1}{x_2-j_1}}\right\} \\
&-(\lambda+\lambda^{\prime}j_3)\sqrt{\frac{-1}{(j_3-j_1)(j_3-j_2)}}\frac{1}{x_2-x_1}\frac{1}{\sqrt{(j_2-j_3)(j_2-j_1)}}\\
&\quad \times \left\{R(x_1)(x_2-j_2)\sqrt{\frac{x_2-j_1}{x_1-j_1}}-R(x_2)(x_1-j_2)\sqrt{\frac{x_1-j_1}{x_2-j_1}}\right\}\\
&=\sqrt{\frac{-1}{(j_3-j_1)(j_1-j_2)}}\frac{1}{x_2-x_1}
\left\{\left(\lambda +\lambda^{\prime}x_2\right)\sqrt{\frac{x_2-j_1}{x_1-j_1}}R(x_1)-\left(\lambda +\lambda^{\prime}x_1\right)\sqrt{\frac{x_1-j_1}{x_2-j_1}}R(x_2)\right\},
\end{align*}
where we have set $R(x)=\sqrt{(j_1-x)(j_2-x)(j_3-x)(j_4-x)(j_5-x)}$. 
Thus, we have 
\begin{align*}
&\sqrt{\frac{j_1-x_2}{j_1-x_1}}\frac{\mathsf{d}x_1}{\mathsf{d}t}+\sqrt{\frac{j_1-x_1}{j_1-x_2}}\frac{\mathsf{d}x_2}{\mathsf{d}t} \\
&\qquad =\frac{2}{x_1-x_2}\left\{\left(\lambda +\lambda^{\prime}x_2\right)R(x_1)\sqrt{\frac{x_2-j_1}{x_1-j_1}}-\left(\lambda +\lambda^{\prime}x_1\right)R(x_2)\sqrt{\frac{x_1-j_1}{x_2-j_1}}\right\}.
\end{align*}
Similar equations can be obtained via cyclic permutations of $(1,2,3)$ for the subindices in $j_1, j_2, j_3$: 
\begin{align*}
&\sqrt{\frac{j_2-x_2}{j_2-x_1}}\frac{\mathsf{d}x_1}{\mathsf{d}t}+\sqrt{\frac{j_2-x_1}{j_2-x_2}}\frac{\mathsf{d}x_2}{\mathsf{d}t} \\
&\qquad =\frac{2}{x_1-x_2}\left\{\left(\lambda +\lambda^{\prime}x_2\right)R(x_1)\sqrt{\frac{x_2-j_2}{x_1-j_2}}-\left(\lambda +\lambda^{\prime}x_1\right)R(x_2)\sqrt{\frac{x_1-j_2}{x_2-j_2}}\right\}, \\
&\sqrt{\frac{j_3-x_2}{j_3-x_1}}\frac{\mathsf{d}x_1}{\mathsf{d}t}+\sqrt{\frac{j_3-x_1}{j_3-x_2}}\frac{\mathsf{d}x_2}{\mathsf{d}t} \\
&\qquad =\frac{2}{x_1-x_2}\left\{\left(\lambda +\lambda^{\prime}x_2\right)R(x_1)\sqrt{\frac{x_2-j_3}{x_1-j_3}}-\left(\lambda +\lambda^{\prime}x_1\right)R(x_2)\sqrt{\frac{x_1-j_3}{x_2-j_3}}\right\}.
\end{align*}
Since $j_1, j_2, j_3$ (and hence $\displaystyle \sqrt{\frac{j_1-x_2}{j_1-x_1}}$, $\displaystyle \sqrt{\frac{j_2-x_2}{j_2-x_1}}$, $\displaystyle \sqrt{\frac{j_3-x_2}{j_3-x_1}}$) can be assumed to be distinct, this means that the following two linear forms in $\left(\xi,\eta\right)$ are the same: 
\begin{align*}
&\frac{\mathsf{d}x_1}{\mathsf{d}t}\xi+\frac{\mathsf{d}x_2}{\mathsf{d}t}\eta, \\
&\frac{2}{x_1-x_2}\left\{\left(\lambda +\lambda^{\prime}x_2\right)R(x_1)\xi-\left(\lambda +\lambda^{\prime}x_1\right)R(x_2)\eta\right\}.
\end{align*}
Therefore, the flow induced by the family of Hamiltonian systems $\lambda C_3+\lambda' C_4$ yields 
\begin{align*}
\frac{\mathsf{d}x_1}{\mathsf{d}t}=2\frac{1}{x_1-x_2}R(x_1)(\lambda^{\prime}x_2+\lambda), \\
\frac{\mathsf{d}x_2}{\mathsf{d}t}=2\frac{1}{x_2-x_1}R(x_2)(\lambda^{\prime}x_1+\lambda), 
\end{align*}
and hence 
\begin{equation}\label{diagonalization}
\begin{cases}
\displaystyle \frac{1}{R(x_1)}\frac{\mathsf{d}x_1}{\mathsf{d}t}+\frac{1}{R(x_2)}\frac{\mathsf{d}x_2}{\mathsf{d}t}
&=-2\lambda^{\prime}, \vspace{2mm}\\
\displaystyle \frac{x_1}{R(x_1)}\frac{\mathsf{d}x_1}{\mathsf{d}t}+\frac{x_2}{R(x_2)}\frac{\mathsf{d}x_2}{\mathsf{d}t}
&=2\lambda.
\end{cases}
\end{equation}
This proves the theorem as the two differential forms $\displaystyle \frac{\mathsf{d}x}{R(x)}$, $\displaystyle \frac{x\mathsf{d}x}{R(x)}$ provide a basis of the space of holomorphic 1-forms on the Jacobian of the hyperelliptic curve $\mathcal{C}_c$. 
\end{proof}

Note that these equations \eqref{diagonalization} induce a uniformization of the Kummer surface \cite[\S 99]{hudson}.

\section{Some special solutions}
In this section, we discuss two types of special solutions to Kirchhoff equations \eqref{KE_K}, \eqref{KE_p} of Clebsch top. 
Namely, as is pointed out in \cite[\S 2.2, \S 2.4]{HJL}, there are three special solutions on the three-dimensional vector spaces $p_1=K_2=K_3=0$, $p_2=K_3=K_1=0$, $p_3=K_1=K_2=0$, as well as the solutions on the three-dimensional vector spaces $p_{\alpha}=\delta_{\alpha}K_{\alpha}$, $\alpha=1,2,3$, with suitable constants $\delta_{\alpha}$, $\alpha=1,2,3$. 
In fact, assuming e.g. $p_1=K_2=K_3$, we can reduce Kirchhoff equations \eqref{KE_K}, \eqref{KE_p} to 
\begin{equation}\label{special_sol_1}
\begin{cases}
\displaystyle \frac{\mathsf{d}K_1}{\mathsf{d}t}&=\displaystyle \left(\frac{1}{m_3}-\frac{1}{m_2}\right)p_2p_3, \vspace{1mm}\\
\displaystyle \frac{\mathsf{d}p_2}{\mathsf{d}t}&=\displaystyle \frac{1}{I_1}p_3K_1, \vspace{1mm}\\
\displaystyle \frac{\mathsf{d}p_3}{\mathsf{d}t}&=\displaystyle -\frac{1}{I_1}K_1p_2, 
\end{cases}
\end{equation}
while, under the condition $p_{\alpha}=\delta_{\alpha}K_{\alpha}$, $\alpha=1,2,3$, the equations \eqref{KE_K}, \eqref{KE_p} are transformed to 
\begin{equation}\label{special_sol_2}
\begin{cases}
\displaystyle \frac{\mathsf{d}K_1}{\mathsf{d}t}&=\displaystyle \frac{1}{\delta_1}\left(\frac{\delta_2}{I_3}-\frac{\delta_3}{I_2}\right)K_2K_3, \vspace{1mm}\\
\displaystyle \frac{\mathsf{d}K_2}{\mathsf{d}t}&=\displaystyle \frac{1}{\delta_2}\left(\frac{\delta_3}{I_1}-\frac{\delta_1}{I_3}\right)K_3K_1, \vspace{1mm}\\
\displaystyle \frac{\mathsf{d}K_3}{\mathsf{d}t}&=\displaystyle \frac{1}{\delta_3}\left(\frac{\delta_1}{I_2}-\frac{\delta_2}{I_1}\right)K_1K_2.
\end{cases}
\end{equation}

To describe the Hamiltonian structures of these equations \eqref{special_sol_1} and \eqref{special_sol_2}, we first consider the modified Lie bracket $\left[\cdot, \cdot\right]_M$ on $\mathbb{C}^3$ associated to a diagonal matrix $M=\mathrm{diag}\left(\mu_1, \mu_2, \mu_3\right)$: 
\[
\left[\bm{u}, \bm{v}\right]_M:=M\left(\bm{u}\times \bm{v}\right), \bm{u}, \bm{u}\in\mathbb{C}^3. 
\]
In association to the modified Lie bracket $\left[\cdot, \cdot\right]_M$, we can define Lie-Poisson bracket $\left\{\cdot, \cdot\right\}_M$ through 
\[
\left\{F, G\right\}_M(\bm{x}):=\left\langle \bm{x}, M\left(\nabla F(\bm{x})\times \nabla G(\bm{x})\right)\right\rangle, 
\]
where $\bm{x}\in\mathbb{C}^3$ and $F, G$ are differentiable functions on $\mathbb{C}^3$. 
The Lie bracket $\left[\cdot, \cdot\right]_M$ and the associated Lie-Poisson bracket $\left\{\cdot, \cdot\right\}_M$ are considered in \cite{holm_marsden_1991} and used to relate the free rigid body and the simple pendulum. 
(It is further used in \cite{iwai_tarama_2010} to analyze the quantization problems.) 

Since $\left\{F, G\right\}_M=\left\langle (M\bm{x})\times \nabla F(\bm{x}), \nabla G(\bm{x})\right\rangle$, the Hamiltonian vector field $\Xi_F^{(M)}$ for the Hamiltonian $F$ is given as $\left(\Xi_F^{(M)}\right)_{\bm{x}}=(M\bm{x})\times \nabla F(\bm{x})$. 
In particular, if we choose the Hamiltonian $F$ as $F(\bm{x})=\displaystyle \frac{1}{2}\left\langle \bm{x}, \mathrm{diag}\left(f_1, f_2,f_3\right)\bm{x}\right\rangle$, then the Hamiltonian vector field is written as 
\[
\left(\Xi_F^{(M)}\right)_{\bm{x}}=
\left(\left(\mu_2 f_3-\mu_3 f_2\right)x_2x_3,. \left(\mu_3 f_1-\mu_1 f_3\right)x_3x_1, \left(\mu_1 f_2-\mu_2 f_1\right)x_1x_2\right)^{\mathrm{T}},
\]
where $\bm{x}=\left(x_1, x_2, x_3\right)^{\mathrm{T}}$. 

Then, one can easily check that, regarding $x_1=K_1$, $x_2=p_2$, $x_3=p_3$, we can recover \eqref{special_sol_1}, by setting $\mu_{\alpha}=\displaystyle \frac{1}{I_{\alpha}}$, $\alpha=1,2,3$, $f_1=0$, $f_2=-1$, $f_3=-1$. 
On the other hand, \eqref{special_sol_2} can be obtained for $x_{\alpha}=K_{\alpha}$, $\alpha=1,2,3$, by substituting as $\mu_1=\displaystyle \frac{\delta_1}{\delta_2\delta_3}$, $\mu_2=\displaystyle \frac{\delta_2}{\delta_3\delta_1}$, $\mu_3=\displaystyle \frac{\delta_3}{\delta_1\delta_2}$, $f_{\alpha}=\displaystyle \frac{\delta_{\alpha}}{I_{\alpha}}$, $\alpha=1,2,3$. 

The solutions can be given in terms of elliptic functions because there are two quadratic first integrals $\displaystyle \frac{1}{2}\left\langle \bm{x}, M\bm{x}\right\rangle$, $F(\bm{x})=\displaystyle \frac{1}{2}\left\langle \bm{x}, \mathrm{diag}\left(f_1, f_2,f_3\right)\bm{x}\right\rangle$ and because the intersection of the level surfaces of these functions are elliptic curves. 
This is a parallel argument to the one for Euler top. 

\medskip

The above two kinds of special solutions can be found in the intersection of four quadric level surfaces for the functions $C_1$, $C_2$, $C_3$, $C_4$ in \eqref{four_first_int}, but with specific values of $C_3$ and $C_4$. 
(Recall that we assume $C_1=0$ and $C_2=1$.) 
About the solution \eqref{special_sol_1}, the condition $p_1=K_2=K_3=0$ implies that $C_1\equiv 0$, $C_2=p_2^2+p_3^2=1$, $C_3=K_1^2+(j_3+j_1)p_2^2+(j_1+j_2)p_3^2$, $C_4=j_1K_1^2+j_3j_1p_2^2+j_1j_2p_3^2$. 
Now, it is easy to check that $j_1^2-C_3j_1^2+C_4=0$. 
Taking \eqref{equation_c_3_c_4} into account, we see that this is equivalent to the case where either $j_1=j_4$ or $j_1=j_5$ is satisfied. 
Clearly, the same type of special solution with the condition $p_2=K_3=K_1=0$ can be obtained only if $j_2^2-C_3j_2+C_4=(j_2-j_4)(j_2-j_5)=0$ and the one with $p_3=K_1=K_2$ can be only if $j_3^2-C_3j_3+C_4=(j_3-j_4)(j_3-j_5)=0$. 

About the solution \eqref{special_sol_2}, the values of $C_3$ and $C_4$ must satisfy the condition $j_4=j_5$. 
In fact, if we assume the condition $p_{\alpha}=\delta_{\alpha}K_{\alpha}$, $\alpha=1,2,3$, the functions $C_1$, $C_2$, $C_3$, $C_4$ in \eqref{four_first_int} amount to 
\begin{align}\label{condition_p_K}
C_1&=\displaystyle \sum_{\alpha=1}^3\delta_{\alpha}K_{\alpha}^2, \quad
C_2=\displaystyle \sum_{\alpha=1}^3\delta_{\alpha}^2K_{\alpha}^2, \notag \\
C_3&=\displaystyle \sum_{\alpha=1}^3\left(1+\left(j_{\beta}+j_{\gamma}\right)\delta_{\alpha}^2\right)K_{\alpha}^2, \\
C_4&=\displaystyle \sum_{\alpha=1}^3\left(j_{\alpha}+j_{\beta}j_{\gamma}\delta_{\alpha}\right)K_{\alpha}^2. \notag
\end{align}
Here, $\left\{\alpha, \beta, \gamma\right\}=\left\{1,2,3\right\}$. 
For these functions, we have the dependence relation $C_3^2-4C_4=0$, which is equivalent to say that the quadric equation \eqref{equation_c_3_c_4} has one double root, namely $j_4=j_5$. 

To verify the relation $C_3^2-4C_4=0$, we first show the following lemma. 
\begin{lemma}
For the condition $p_{\alpha}=\delta_{\alpha}K_{\alpha}$, $\alpha=1,2,3$, to be satisfied along an integral curve of the Hamiltonian flow for the Hamiltonian $\lambda C_3+\lambda^{\prime}C_4$, it is necessary that there exist two constants $\sigma, \sigma^{\prime}\in\mathbb{C}$ such that $\left(j_1-\sigma^{\prime}\right)\left(j_2-\sigma^{\prime}\right)\left(j_3-\sigma^{\prime}\right)=\sigma^2$ and that $j_{\alpha}=\sigma\delta_{\alpha}+\sigma^{\prime}$, $\alpha=1,2,3$. 
\end{lemma}
\begin{proof}
For $p_{\alpha}=\delta_{\alpha}K_{\alpha}$, $\alpha=1,2,3$ to be invariant along an integral curve of $\Xi_{\lambda C_3+\lambda^{\prime}C_4}$, it is necessary and sufficient that $\left\{\lambda C_3+\lambda^{\prime}C_4, p_{\alpha}-\delta_{\alpha}K_{\alpha}\right\}=0$, $\alpha=1,2,3$. 
Since 
\begin{align*}
&\left\{\lambda C_3+\lambda^{\prime}C_4, p_1-\delta_1K_1\right\} \\
&\quad =
\left\{\left(\lambda+\lambda^{\prime}j_3\right)\left(\delta_1-\delta_2\right)-\left(\lambda+\lambda^{\prime}j_2\right)\left(\delta_1-\delta_2\right)+\left(j_2-j_3\right)\left(\lambda+\lambda^{\prime}j_1\right)\delta_1\delta_2\delta_3\right\}K_2K_3, 
\end{align*}
we have $\left(\lambda+\lambda^{\prime}j_3\right)\left(\delta_1-\delta_2\right)-\left(\lambda+\lambda^{\prime}j_2\right)\left(\delta_1-\delta_2\right)+\left(j_2-j_3\right)\left(\lambda+\lambda^{\prime}j_1\right)\delta_1\delta_2\delta_3=0$. 
By cyclic change of parameters, we obtain the system of linear equations in $\left(\delta_1-\delta_2, \delta_1-\delta_3, \delta_1\delta_2\delta_3\right)^{\mathrm{T}}$: 
\begin{equation*}
\begin{pmatrix}
\lambda+j_3\lambda^{\prime} & -\left(\lambda+j_2\lambda^{\prime}\right) & \left(j_2-j_3\right)\left(\lambda+j_1\lambda^{\prime}\right) \\
\left(j_3-j_1\right)\lambda^{\prime} & \lambda +j_1\lambda^{\prime} & \left(j_3-j_1\right)\left(\lambda+j_1\lambda^{\prime}\right) \\
-\left(\lambda +j_1\lambda^{\prime}\right) & \left(j_1-j_2\right)\lambda^{\prime} & \left(j_1-j_2\right)\left(\lambda+j_3\lambda^{\prime}\right)
\end{pmatrix}
\begin{pmatrix}
\delta_1-\delta_2\\ \delta_1-\delta_3\\ \delta_1\delta_2\delta_3
\end{pmatrix}
=\begin{pmatrix}
0 \\ 0 \\ 0
\end{pmatrix}. 
\end{equation*}
These equations have already appeared in \cite[\S A.1]{HJL} and it can be reduced through Gau{\ss} method to the simple equations 
\begin{equation*}
\begin{pmatrix}
1 & 0 & -\left(j_1-j_2\right) \\ 
0 & 1 & -\left(j_1-j_3\right) \\ 
0 & 0 & 0
\end{pmatrix}
=\begin{pmatrix}
0 \\ 0 \\ 0
\end{pmatrix}. 
\end{equation*}
Thus, there exist $s, s^{\prime}\in\mathbb{C}$ such that $\delta_{\alpha}=sj_{\alpha}+s^{\prime}$. 
In other words, we have $j_{\alpha}=\sigma\delta_{\alpha}+\sigma^{\prime}$, $\alpha=1,2,3$, where $\sigma=1/s$ and $\sigma^{\prime}=-s^{\prime}/s$. 
By the condition $\delta_1\delta_2\delta_3=1$, we have the condition $\left(sj_1+s^{\prime}\right)\left(sj_2+s^{\prime}\right)\left(sj_3+s^{\prime}\right)=s\iff \left(j_1-\sigma^{\prime}\right)\left(j_2-\sigma^{\prime}\right)\left(j_3-\sigma^{\prime}\right)=\sigma^2$. 
\end{proof}
Inserting the condition $j_{\alpha}=\sigma\delta_{\alpha}+\sigma^{\prime}$, $\alpha=1,2,3$, to \eqref{condition_p_K}, we have 
\begin{align*}
C_3
&=\sum_{\alpha=1}^3\left\{1+\sigma\left(\delta_{\beta}+\delta_{\gamma}\right)\delta_{\alpha}^2+2\sigma^{\prime}\delta_{\alpha}^2\right\}K_{\alpha}^2 \\
&=\frac{1}{s}\sum_{\alpha=1}^3\left\{s+\left(\delta_{\alpha}\delta_{\beta}+\delta_{\gamma}\delta_{\alpha}\right)\delta_{\alpha}\right\}K_{\alpha}^2+2\sigma^{\prime}C_2 \\
&=\sigma\sum_{\alpha=1}^3\left\{\delta_{\alpha}\delta_{\beta}\delta_{\gamma}+\left(\delta_{\gamma}\delta_{\alpha}+\delta_{\alpha}\delta_{\beta}\right)\delta_{\alpha}\right\}K_{\alpha}^2+2\sigma^{\prime}C_2 \\
&=\sigma\sum_{\alpha=1}^3\left(\delta_{\beta}\delta_{\gamma}+\delta_{\alpha}\delta_{\beta}+\delta_{\gamma}\delta_{\alpha}\right)\delta_{\alpha}K_{\alpha}^2+2\sigma^{\prime}C_2 \\
&=\sigma\sum_{\alpha=1}^3\left(\delta_1\delta_2+\delta_2\delta_3+\delta_3\delta_1\right)\delta_{\alpha}K_{\alpha}^2+2\sigma^{\prime}C_2 \\
&=\left(\delta_1\delta_2+\delta_2\delta_3+\delta_3\delta_1\right)\sigma C_1+2\sigma^{\prime}C_2=2\sigma^{\prime} 
\end{align*}
and similarly 
\begin{align*}
C_4
&=\sum_{\alpha=1}^3\left\{\sigma\delta_{\alpha}+\sigma^{\prime}+\sigma^2\delta_{\beta}\delta_{\gamma}\delta_{\alpha}^2+\sigma\sigma^{\prime}\left(\delta_{\beta}+\delta_{\gamma}\right)\delta_{\alpha}^2+\left(\sigma^{\prime}\right)^2\delta_{\alpha}^2\right\}K_{\alpha}^2\\
&=\sigma C_1+\left(\sigma^{\prime}\right)^2C_2+\sigma^2\delta_1\delta_2\delta_3\sum_{\alpha=1}^3\delta_{\alpha}K_{\alpha}^2+\sum_{\alpha=1}^3\left\{\sigma+\left(\delta_{\beta}+\delta_{\gamma}\right)\delta_{\alpha}^2\right\}K_{\alpha}^2 \\
&=\left\{2\sigma -\sigma\sigma^{\prime}\left(\delta_1\delta_2+\delta_2\delta_3+\delta_3\delta_1\right)\right\}C_1+\left(\sigma^{\prime}\right)^2C_2
=\left(\sigma^{\prime}\right)^2. 
\end{align*}
Here we used the condition $C_1=0$ and $C_2=1$. 
Therefore, we have $C_3^2-4C_4=0$. 

\section{Explicit computation of action-angle coordinates}

The Clebsch case displays many similarities with the Kowalevski top as there is a linearization on the Jacobian of a genus two curve. Indeed the explicit computation of the action-angle coordinates can be made quite similarly with that made for the Kowalevski top \cite{Fran} at least for the Weber case $(C_1=0)$.
In this paragraph, we return back to the real coordinates, since the action-angle coordinates are usually considered for real completely integrable Hamiltonian systems. 

We first include a detailed report of the general result proved in \cite{Fran}.

We consider an integrable Hamiltonian system with $m$ degrees of freedom, defined on a symplectic manifold $(V^{2m},\omega)$ of dimension $2m$ and where $\omega$ is the symplectic form, as the data of $m$ generically independent functions:
$$f_j:V^{2m}\to\mathbb{R}, j=1, \cdots, m,$$
such that $\{f_i,f_j\}=0$, where $\{\cdot,\cdot\}$ denotes the Poisson bracket associated with the symplectic form $\omega$.

For instance, the system that we have considered, in restriction to the $4$-dimensional symplectic leaf $V^4: C_1=0, C_2=1$ defines with the couple $C_3,C_4$ an integrable Hamiltonian system with $m=2$ degrees of freedom.

We assume furthermore that $f=(f_1, \ldots, f_m):V^{2m}\to \mathbb{R}^m$ is proper. Denote $D$ the critical locus of $f$, then for $c\in\mathbb{R}^m-D$, the connected components of $f^{-1}(c)$ are $m$-dimensional real tori.

Inspired by several examples (including the Kowalevski top), the following ``algebraic linearization" condition $(A)$ was introduced in \cite{Fran}:

\begin{enumerate}
\item 
There exists $m$ generically independent functions $(x_1, \ldots, x_m)$ defined on $V^{2m}$ so that, together with the collection $(f_1, \ldots, f_m)$ of first integrals, we obtain a system of coordinates on an open dense set of $V^{2m}-f^{-1}(D)$.
\item 
There is a family of hyperelliptic curves of genus $m$, $y^2=P(x,c)$, where $P$ is a polynomial in $x$ of degree $2m$ or $2m+1$, parameterized by regular values $c\in \mathbb{R}^m\setminus D$ of $\left(f_1, \ldots, f_m\right)$, so that
\begin{equation}
\label{A}
\sum_{k=1}^m \frac{x_k^{j-1}\{f_i,x_k\}}{\sqrt{P(x_k, c)}}=W_{ij},
\end{equation}
where the $m\times m$ matrix $W=\left(W_{ij}\right)$ is  invertible and constant.
\end{enumerate}
In \cite{Fran}, it was shown that the Kowalevski top satisfies the condition $(A)$ with the constant invertible matrix $W=\mathrm{diag}(1,2)$. 
We can easily check, from the above computations, that this condition is also verified in the Clebsch case, with $C_1=0$.

The following was proved in \cite{Fran}:
\begin{theorem}
\label{preparation of the symplectic form}
Assume that an integrable Hamiltonian system satisfies the condition $(A)$, then the symplectic form $\omega$ can be written as:
\begin{align*}
\omega&=\sum_{l=1}^m \eta_l\wedge \mathsf{d}f_l,\\
\eta_l&=\sum_{j=1}^m A_{jl}\mathsf{d}f_j+\sum_{j=1}^mB_l(x_j)\mathsf{d}x_j,\\
B_l(x)&=\sum_{k=1}^m\frac{(W^{-1})_{kl}x^{k-1}}{\sqrt{P(x,c)}}.
\end{align*}
\end{theorem}
Consider a family of cycles $\gamma_j(c), (j=1, \cdots, m)$ which defines a system of generators of the homology of the torus $f^{-1}(c)$, where $c\in \mathbb{R}^m\setminus D$ are regular values. 
We introduce the so-called ``period matrix" $\Psi=\left(\Psi_{ij}\right)$ given by:
\begin{equation}
\Psi_{ij}(c)=\int_{\gamma_j(c)}\eta_i.
\end{equation}

As the $1$-forms $\eta_i$ are closed when restricted to the tori $f^{-1}(c)$, it is easy to check that the $\Psi_{ij}(c)$ depend only on the homology class of $\gamma_j(c)$ in $H_{1}(f^{-1}(c), \mathbb{Z})$. 
By the previous theorem, we obtain:
\begin{equation*}
\Psi_{ij}(c)=\int_{\gamma_j(c)}\sum_{l,k}\frac{x_l^{k-1}(W^{-1})_{ki}}{\sqrt{P(x_l,c)}}\mathsf{d}x_l.
\end{equation*}

It is known that the curves $\mathcal{C}_c: y^2=P(w,c)$ is mapped injectively into its Jacobian 
$\mathrm{Jac}(\mathcal{C}_c)=H^0(\mathcal{C}_c, \Omega^1_{\mathcal{C}_c})^*/H_1(\mathcal{C}_c, \mathbb{Z})$ via the Abel map $S^m(\mathcal{C}_c)\to \mathrm{Jac}(\mathcal{C}_c)$ from the symmetric product $S^m(\mathcal{C}_c)$ of $\mathcal{C}_c$ and that this injection is a quasi-isomorphism (induces an isomorphism at the level of homology). 
So we can assume that the generators $\gamma_j(c)$ are already given as paths on the curve $\mathcal{C}_c$. As consequence, we can suppress the index $l$ in the above formula and write:
\begin{equation}
\label{period matrix in case (A)}
\Psi_{ij}(c)=\int_{\gamma_j(c)}\sum_{k}\frac{x^{k-1}(W^{-1})_{ki}}{\sqrt{P(x,c)}}\mathsf{d}x.
\end{equation}
Note that the Riemann theorem on the usual period matrix of hyperelliptic surfaces yields the invertibility of the ``period matrix" $\Psi(c)$. 
Assume for simplicity that the symplectic form $\omega$ is exact, $\omega=\mathsf{d}\eta$, then the actions $a_j, j=1, \cdots, m$ are defined as functions of $c$ (or equivalently of the functions $f_1, \ldots, f_m$) by 
\begin{equation*}
a_j=f^*(a_j(c)), \qquad a_j(c)=\int_{\gamma_j(c)}\eta,
\end{equation*}
and thus the derivatives are given as 
\begin{equation}
\label{the actions}
\frac{\partial a_j}{\partial f_i}=\Psi_{ij}(c), \qquad i,j=1, \cdots, m.
\end{equation}

Clearly, the formula \eqref{period matrix in case (A)} can be applied to the Clebsch case (with $C_1=0$) that we have been studying. 
In that case, the supplementary coordinates $(x_1,x_2)$ are in involution with respect to the Poisson bracket $\{\cdot,\cdot\}$. 
In this case, we have $P(x,c)=\sqrt{(j_1-x)(j_2-x)(j_3-x)(x^2-C_3x+C_4)}$, whose three zeros $j_1, j_2, j_3$ are independent of $C_3, C_4$, which is very different for example from Kowalevski top. 
In the case $j_1<j_2<j_3$ we can choose the two generators $\gamma_1(c),\gamma_2(c)$ associated with the two paths lifted from the complex $x$-plane to the curve $\mathcal{C}_c$ obtained for instance as a path which turns around $j_1$ and $j_2$ and another path which encircles $j_3$ and $j_4$. The formula \eqref{period matrix in case (A)} yields now a very simple formula 
\begin{align*}
\Psi_{i1}&=\int_{j_1}^{j_2}\left\{\left(W^{-1}\right)_{1i}\frac{\mathsf{d}x}{P(x,c)}+\left(W^{-1}\right)_{2i}\frac{x\mathsf{d}x}{P(x,c)}\right\}, \\
\Psi_{i2}&=\int_{j_3}^{j_4}\left\{\left(W^{-1}\right)_{1i}\frac{\mathsf{d}x}{P(x,c)}+\left(W^{-1}\right)_{2i}\frac{x\mathsf{d}x}{P(x,c)}\right\},
\end{align*}
where the matrix $W$ is diagonal (cf. \eqref{diagonalization}): 
$W=
\begin{pmatrix}
-2 & 0 \\ 0 & 2
\end{pmatrix}$. 
Thus, we have 
\begin{align*}
\Psi_{11}&=\int_{j_1}^{j_2} -2\frac{x\mathsf{d}x}{R(x,c)}, \\
\Psi_{21}&=\int_{j_1}^{j_2} 2\frac{\mathsf{d}x}{R(x,c)}, \\
\Psi_{12}&=\int_{j_3}^{j_4} -2\frac{\mathsf{d}x}{R(x,c)}, \\
\Psi_{22}&=\int_{j_3}^{j_4} 2\frac{x\mathsf{d}x}{R(x,c)}. 
\end{align*}
The integration of this ``period matrix" is quite easy and it yields 
\begin{align*}
a_1=-2\int_{j_1}^{j_2}\sqrt{\frac{x^2-C_3x+C_4}{(x-j_1)(x-j_2)(x-j_3)}}\mathsf{d}x, \\
a_2=-2\int_{j_3}^{j_4}\sqrt{\frac{x^2-C_3x+C_4}{(x-j_1)(x-j_2)(x-j_3)}}\mathsf{d}x.
\end{align*}

\section{Conclusion and Perspectives}

As it has been shown in this article, the Clebsch case under Weber's condition (i.e. $C_1=0$) provides a quite remarkable example of the deep connections between integrable systems and algebraic geometry. 
By simple elimination of variables from the quadratic equations given by the constant of motions, it is possible to deduce naturally an associated Kummer surface. 
It is known in general that there always exists a double covering of a Kummer surface which is the Jacobian of an hyperelliptic curve of genus two. 
In the case of Clebsch top under Weber's condition ($C_1=0$), such a curve can be derived from the induced Hamiltonian Dynamics. 
It is an interesting perspective to further develop the analysis of the actions over the singularities of the system and the connection to a global aspects such as monodromy. 
Another important issue would be the extension to the general case ($C_1\neq 0$), as treated by K\"{o}tter \cite{kot}. 

\medskip

\noindent\textbf{Acknowledgment:} 
The second author thank Jun-ichi Matsuzawa for the stimulating and valuable discussion about Kummer surfaces.


\begin{thebibliography}{100}
\bibitem{admo84}
M. Adler and P. van Moerbeke, Geodesic Flow on $so(4)$ and the Intersection of Quadrics, Proc. Natl. Acad. Sci. USA, \textbf{81}, 4613-4616, 1984.

\bibitem{admo87}
M. Adler and P. van Moerbeke, The intersection of Four Quadrics in $\mathbb{P}^6$. Abelian surfaces and their Moduli, Math. Ann., \textbf{279}, 25-85, 1987.

\bibitem{aom}
Kazuhiko Aomoto, Rigid body motion in a perfect fluid and $\theta$-formulae \--- classical works by H. Weber \--- (in Japanese), RIMS K\^{o}ky\^{u}roku, \textbf{414}, 98-114, 1981. 

\bibitem{arnold_1989}
V. I. Arnol'd, \textit{Mathematical Methods of Classical Mechanics}, 2nd ed., Springer-Verlag, New York-Tokyo, 1989. 

\bibitem{audin_1996}
M. Audin, {\it Spinning Tops}, Cambridge University Press, Cambridge, 1996. 

\bibitem{barth-hulek-peters-vandeven}
W. Barth, K. Hulek, C. Peters, and A. Van de Ven, \textit{Compact Complex Surfaces}, 2nd ed., Springer-Verlag, 2004.

\bibitem{beauville}
A. Beauville, Syst\`emes Hamiltoniens compl\'etement int\'egrables associ\'es aux surfaces K3,  Problems in the Theory of surfaces and their classification (Cortona, 1988) Symp. Math. vol. 32 , Academic Press London, 25-31, 1991.

\bibitem{EFe}
Z. Enolsky and Yu. N. Fedorov, Algebraic description of Jacobians isogenous to Certain Prym Varieties with polarization (1,2). Experimental Mathematics, \textbf{112}, 2016.

\bibitem{fischer}
G. Fischer ed., \textit{Mathematical Models}, 2nd ed., \textbf{Springer Spektrum}, Springer Nature, Wiesbaden, 2017. 

\bibitem{Fran}
J.-P. Fran\c{c}oise, Calcul explicite d'Action-Angles, in \textit{Syst\`emes dynamiques non lin\'eaires: int\'egrabilit\'e et comportement qualitatif}, S\'em. Math. Sup., \textbf{102}, Presses Univ. Montr\'eal, Montreal, QC, 101–-120, 1986. 

\bibitem{francoise_tarama_2015}
J.-P. Fran\c{c}oise and D. Tarama, Analytic extension of the Birkhoff normal forms for the free rigid body dynamics on $SO(3)$, Nonlinearity, \textbf{28}, 1193-1216, 2015.  

\bibitem{haine}
L. Haine, Geodesic Flow on $so(4)$ and Abelian surfaces, Math. Ann., \textbf{263}, 435-472, 1983.

\bibitem{holm_marsden_1991}
D. D. Holm and J. E. Marsden, The rotor and the pendulum, in \textit{Symplectic Geometry and Mathematical Physics}, P. Donato et al. eds., Birkh\"{a}user, Bost-Basel, 189-203, 1991. 

\bibitem{HJL}
P. Holmes, J. Jenkins and N. E. Leonard,
Dynamics of the Kirchhoff equations I: coincident centers of gravity and buoyancy, Physica D, \textbf{118}, 311-342, 1998.

\bibitem{hudson}
R. W. H. T. Hudson, \textit{Kummer's quartic Surfaces}, Cambridge University Press, Cambridge-New York-Port Chester-Melbourne-Sydney, 1990.

\bibitem{iwai_tarama_2010}
T. Iwai and D. Tarama, Classical and quantum dynamics for an extended free rigid body, Diff. Geom. Appl., \textbf{28}, 501-517, 2010.

\bibitem{kot}
F. K\"{o}tter, Ueber die Bewegung eines festen K\"{o}pers in einer Fl\"{u}ssigkeit, I, II, J. Reine Angew. Math., \textbf{109}, 51-81, 89-111, 1892.

\bibitem{kummer}
E. Kummer, \"{U}ber die Fl\"{a}chen vierten Grades mit schzehn singul\"{a}ren Punkten, Monatsberichte der K\"{o}niglichen Preussische Akademie des Wissenschaften zu Berlin, 246-260, 1864. 

\bibitem{MaS}
F. Magri and T. Skrypnyk, The Clebsch System, preprint, arXiv:1512.04872, 2015.

\bibitem{markushevich}
D. Markushevich, Some Algebro-Geometric Integrable Systems Versus classical ones,
CRM proceedings and lecture notes, \textbf{32},197-218, 2002. 

\bibitem{marsden-ratiu}
J. E. Marsden and T. S. Ratiu, \textit{Introduction to Mechanics and Symmetry}, 2nd ed., Springer, New York, 2010.
 
\bibitem{mukai}
S. Mukai, Symplectic structure of the moduli of sheaves on an Abelian or K3 surface, Invent. Math., \textbf{77}, 101-116, 1984.

\bibitem{naruki_tarama_2011}
I. Naruki and D. Tarama, Algebraic geometry of the eigenvector mapping for a free rigid body, Diff. Geom. Appl., \textbf{29} Supplement 1, S170-S182, 2011. 

\bibitem{naruki_tarama_2012}
I. Naruki and D. Tarama, Some elliptic fibrations arising from free rigid body dynamics, Hokkaido Math. J., \textbf{41}(3), 365-407, 2012.

\bibitem{nikulin_1975}
V. Nikulin, On Kummer surfaces, Math. USSR. Izvestija, \textbf{9} (2), 261-275, 1975. 

\bibitem{ratiu-et-al}
T. S. Ratiu et al., A Crash Course in Geometric Mechanics, in: \textit{Geometric Mechanics and Symmetry: the Peyresq Lectures}, J. Montaldi and T. Ratiu (eds.), Cambridge University Press, Cambridge, 2005.

\bibitem{schottky_1891} 
F. Schottky, \"{U}ber das analytische Problem der Rotation eines starren K\"{o}rpers im Raume von vier Dimensionen, Sitzungber. K\"{o}nig. Preuss. Akad. Wiss. zu Berlin, \textbf{27}, 227-232, 1891. 

\bibitem{tarama_francoise_2014}
J.-P. Fran\c{c}oise and D. Tarama, Analytic extension of Birkhoff normal forms for Hamilotnian systems of one degree of freedom \--- Simple pendulum and free rigid body dynamics \---, RIMS K\^{o}ky\^{u}roku Bessatsu, \textbf{B52}, 219-236, 2014. 

\bibitem{web1}
H. Weber, \"{U}ber die Kummersche Fl\"{a}che vieter-Ordnung mit sechzehn Knoterpunkten und ihre Beziehung zu den Thetafunctionen mit zwei Veranderlichen, J. Reine Angew. Math.  \textbf{84}, 332-354, 1878.

\bibitem{web2}
H. Weber, Anwendung der Thetafunctionen zweier Veranderlichen auf die theorie der Bewegung eines festen Korpers in einer Flussingkeit, Math. Ann., \textbf{14}, 173-206, 1879.
\end{thebibliography}
\end{document}